\documentclass{revtex4}

\usepackage{amsmath}
\usepackage{amssymb}
\usepackage{amsfonts}
\usepackage{latexsym}
\usepackage{bm}
\usepackage{color}
\usepackage{theorem}

\catcode `\@=11 \@addtoreset{equation}{section}

\catcode `\@=12

\newcommand{\be}{\begin{equation}}
\newcommand{\en}{\end{equation}}
\newcommand{\bea}{\begin{eqnarray}}
\newcommand{\ena}{\end{eqnarray}}
\newcommand{\beano}{\begin{eqnarray*}}
\newcommand{\enano}{\end{eqnarray*}}
\newcommand{\bee}{\begin{enumerate}}
\newcommand{\ene}{\end{enumerate}}

\newcommand{\R}{\mathbb{R}}
\newcommand{\mc}{\mathcal}

\newcommand{\D}{{\mc D}}

\newcommand{\Sc}{{\cal S}}
\newcommand{\E}{{\cal E}}
\newcommand{\F}{{\cal F}}

\newcommand{\Lc}{{\cal L}}
\newcommand{\C}{{\cal C}}
\newcommand{\1}{1 \!\! 1}

\newcommand{\Hil}{\mc H}

\newtheorem{thm}{Theorem}[section]
\newtheorem{cor}[thm]{Corollary}
\newtheorem{lemma}[thm]{Lemma}
\newtheorem{prop}[thm]{Proposition}
\newtheorem{defn}[thm]{Definition}

\newenvironment{proof}{\noindent {\bf Proof --}}{\hfill$\square$ \vspace{3mm}\endtrivlist}

\newcommand{\ip}[2]{\left\langle {#1},{#2} \right\rangle}
\newcommand{\balpha}{{\mbox{\boldmath${\alpha}$}}}

\catcode `\@=11 \@addtoreset{equation}{section}
\catcode `\@=12

\begin{document}

\title{Biorthogonal vectors, sesquilinear forms and some physical operators}


\author{F. Bagarello} \affiliation{ Dipartimento di Energia, Ingegneria dell'Informazione e Modelli Matematici,
Facolt\`a di Ingegneria, Universit\`a di Palermo, I-90128  Palermo, and INFN, Sezione di Napoli, ITALY\\
e-mail: fabio.bagarello@unipa.it\,\,\,\, Home page: www1.unipa.it/fabio.bagarello}

\author{H. Inoue} \affiliation{Center for advancing Pharmaceutical Education, Daiichi University of Pharmacy, 22-1 Tamagawa-cho, Minami-ku, Fukuoka 815-8511, Japan\\
e-mail: h-inoue@daiichi-cps.ac.jp}

\author{C.Trapani}\affiliation{Dipartimento di Matematica e Informatica, Universit\`a di Palermo, I-90123 Palermo, Italy\\e-mail: camillo.trapani@unipa.it}

\date{February 17, 2018}

%


\begin{abstract}
\noindent Continuing the analysis undertaken in previous articles, we discuss some features of non-self-adjoint operators and sesquilinear forms which are defined starting from two biorthogonal families of vectors, like the so-called generalized Riesz systems, enjoying certain properties. In particular we discuss what happens when they forms two $\D$-quasi bases.

\end{abstract}
\maketitle

\section{Introduction}

A Riesz basis of a Hilbert space $\Hil$ (with scalar product $\ip{\cdot}{\cdot}$ and related norm $\|.\|$) is a sequence $\{\varphi_n\}$ of elements
 of $\Hil$ that are obtained by transforming  an orthonormal basis $\{e_n\}$ of $\Hil$ by some bounded operator $T$  with bounded inverse; i.e., $\varphi_n=Te_n$, $n\in {\mathbb N}$, \cite{Christensen}.  Every Riesz basis
ia a  {\em frame} \cite{casazza, Christensen, rieszbasis}; i.e., there exist positive numbers $c, C$ such that
\begin{equation}\label{eqn_bounds} c\|x\|^2 \leq \sum_{n=1}^\infty |\ip{x}{\varphi_n}|^2 \leq C \|x\|^2, \quad \forall x \in \Hil.\end{equation}
What makes of a frame a Riesz basis is its exactness: a frame is a Riesz basis if it is no more a frame when
anyone of its elements is dropped out.  The notion of frame is
crucial in signal analysis and for coherent states (see, e.g.
\cite{casazza} and references therein) and it has been extensively explored both from the theoretical point of view and in concrete applications. Moreover, further generalizations have been introduced, with the aim of providing more flexible tools (such as {\em semi-frames} \cite{ant-bal-semiframes1,ant-bal-semiframes2} and {\em reproducing pairs} in \cite{ast,speck-bal}).

The relevance of Riesz bases in physics relies on the fact that sometimes they appear as eigenvectors of non-self-adjoint operators. The simplest situation occurs when an operator $X$ is similar to a self-adjoint operator $H$; i.e., there exists a bounded operator $T$ with bounded inverse $T: D(H)\to D(X)$ and $XTx=THx$, for every $x \in D(H)$. If $H$ has a purely discrete spectrum and $\{e_n\}$ is an orthonormal basis (ONB) of eigenvectors, then the vectors  $\varphi_n=Te_n$ are eigenvectors of $X$ and constitute a Riesz basis for $\Hil$. This situation is very familiar in the so-called Pseudo-Hermitian Quantum Mechanics where the Hamiltonian of a given quantum system is no more required to be a self-adjoint operator.
On the other hand, Riesz bases can be used to define Hamiltonians and lowering and raising operators as in \cite{bit2013}. These reasons have led many authors to try and generalize the notion of Riesz basis mainly by modifying (weakening) the requirements on the operator $T$ or by passing to a different framework such as that of rigged Hilbert spaces \cite{bell_ct}.

Generalized Riesz bases were introduced in \cite{atsushi} and in \cite{hiro_taka}. Here we change the original definition, since it is more convenient for our purposes.

\begin{defn}\label{def11}A sequence $\F_\varphi=\{ \varphi_n \}$ of vectors of $\Hil$ is called a generalized Riesz system  if there exist a densely defined closed operator $T$ in $\Hil$ with densely defined inverse, and an orthonormal basis (ONB) $ \F_{e}= \{e_n\}$ such that $e_n\in D(T) \cap D((T^{-1})^\ast)$
 and $\varphi_n =Te_n$. We call $(\F_e, T)$ a  constructing pair for $\F_\varphi$.
 \end{defn}

\vspace{2mm}

{\bf Remark:--} With respect to what was proposed in \cite{atsushi,hiro_taka, hiro1, hiro2} we prefer to use here the word {\em system} instead of {\em basis}, since the sequence $\{ \varphi_n \}$ need not be a basis of $\Hil$.

\vspace{2mm}
If $\F_\varphi$ is a generalized Riesz system with constructing pair $(\F_{\bm{e}},T)$, then some physical operators can be defined (for example, non-self-adjoint Hamiltonians, lowering, raising and number operators). For this reason, it is important for studies of non-self-adjoint Hamiltonian to consider when generalized Riesz system can be constructed. This continues an analysis originally undertaken by some of us, \cite{bit2013}, and then continued in \cite{btt2016,bb2017} where  the use of biorthogonal sets in connection with {\em physically motivated operators} like Hamiltonians, ladder operators, generalized Gibbs states and intertwining operators has been extensely discussed.

In \cite{hiro1, hiro2} it has been shown that if two biorthogonal sequences $\F_\varphi$ and $\F_\psi$ are regular, that is if the two sets $D_\varphi := $ linear span of $\F_\varphi$ and $D_\psi := $ linear span of $\F_\psi$ are dense in $\Hil$, or if they are even semi-regular, that is, either $D_\varphi$ or $D_\psi$ is dense in $\Hil$, then $\F_\varphi$ and $\F_\psi$ are generalized Riesz systems. In Section 3, we shall consider when $\F_\varphi$ and $\F_\psi$ are generalized Riesz systems in case that $D_\varphi$ and $D_\psi$ are not dense in $\Hil$ by using the theory of positive sesquilinear forms $\Omega_\varphi$ and $\Omega_\psi$ defined as follows:
{if $\{\varphi_n\}$ is a sequence of vectors in $\Hil$,  we introduce the following subspace $D(\varphi)$ of $\Hil$:
\begin{eqnarray} \label{eqn_one}
D(\varphi)
=\left\{ x\in \Hil ; \sum_{k=0}^\infty |\ip{x}{\varphi_k}|^2 < \infty \right\}.
\end{eqnarray}
If $ D(\varphi)$ is dense in $\Hil$,
then a closed positive sesquilinear form $\Omega_\varphi$ on $D(\varphi) \times D(\varphi)$ is defined by
\be
\Omega_\varphi (x,y)
= \sum_{k=0}^\infty \ip{x}{\varphi_k}\ip{\varphi_k}{y}, \quad x,y\in D(\varphi)\nonumber
\en
and so by the representation theory of closed positive sesquilinear forms, \cite{kato, davies}, there exists a unique positive self-adjoint operator $K_\varphi$ in $\Hil$ such that $D(\varphi)=D (K_\varphi^{1/2})$ and $\Omega_\varphi(x,y)=\ip{ K_\varphi^{1/2}x}{K_\varphi^{1/2}y}$, for all $x,y \in D(\varphi)$.

If $\{ \varphi_n\}$ and $\{ \psi_n\}$ are biorthogonal sequences, we can consider the corresponding sesquilinear forms $\Omega_\varphi$ and $\Omega_\psi$, defined in analogy with $\Omega_\varphi$, and, in turn, the corresponding operators $K_\varphi,K_\psi$. The main scope of this paper consists in showing that $\{ \varphi_n\}$ and $\{ \psi_n\}$ are generalized Riesz systems under some technical assumptions on $K_\varphi,K_\psi$ and if they are $\D$-quasi bases, \cite{bag1}, that is, if the following equalities hold
\begin{equation}\label{eqn_quasi_bases}\sum_{k=0}^\infty \ip{x}{\varphi_k}\ip{\psi_k}{y}=\sum_{k=0}^\infty \ip{x}{\psi_k}\ip{\varphi_k}{y}=\ip{x}{y},\quad \forall x,y \in \D.\end{equation} Here $\D$ is a dense subspace in $\Hil$ such that  $\{\varphi_n \} \cup \{ \psi_n\} \subset \D \subset D(\varphi) \cap D(\psi)$.  In Section 4, we shall define the non-self-adjoint Hamiltonians $H_{\varphi,\psi}^{\bm{\alpha}}$ and $H_{\psi,\varphi}^{\bm{\alpha}}$, the generalized lowering operators $A_{\varphi,\psi}^{\bm{\alpha}}$ and $A_{\psi,\varphi}^{\bm{\alpha}}$, and the generalized raising operators $B_{\varphi,\psi}^{\bm{\alpha}}$ and $B_{\psi,\varphi}^{\bm{\alpha}}$ for a generalized Riesz system $\F_\varphi= \{ \varphi_n \}$ for a constructing pair $(\F_{\bm{e}},T)$ and $\{ \alpha_n \} \subseteq \C$, where $\psi_n = (T^{-1})^\ast e_n$, $n=0,1, \ldots$. and investigate when these operators are densely defined and closed.

{ The paper is organized as follows. After some preliminaries (Section 2) we give in Section 3 the main result of this paper consisting in a necessary and sufficient condition for the sequences $\F_\phi$ and $\F_\psi$ to be generalized Riesz systems. In Section 4, in analogy with \cite{bit2013} we discuss some properties of physical operators defined by generalized Riesz systems. Our conclusions are given in Section 5.}



}

\section{Preliminary results on generalized Riesz systems}
{Before going forth, we examine some properties of generalized Riesz systems that can be useful for us.

Following the definition, a generalized Riesz system $\F_\varphi=\{\varphi_n, n\geq0\}$ is constructed by taking the image of an ONB $\F_e=\{e_n, n\geq0\}$ through a densely closed invertible operator $T$ with densely defined inverse. As we have seen in Definition \ref{def11},  $(\F_e, T)$ is a { constructing pair for $\F_\varphi$.} The constructing pair of a generalized Riesz system is not unique.

\begin{prop}\label{addprop1} Let $\{ \varphi_n \}$ be a generalized Riesz system, with constructing pair $(\F_e, T)$ . Then $D(\varphi)= D(T_0^*)$, where $T_0$ denotes the restriction of $T$ to the linear span $D_e$ of the ONB $\{e_n\}$ and $D(\varphi)$ is dense in $\Hil$.
\end{prop}
\begin{proof} Indeed, let $x\in D(T_0^*)$,  Then we have:
$$ \sum_{k=0}^\infty |\ip{x}{\varphi_k}|^2 =\sum_{k=0}^\infty |\ip{x}{Te_k}|^2= \sum_{k=0}^\infty |\ip{T_0^*x}{e_k}|^2=\|T_0^*x\|^2<\infty.$$ On the other hand, if $x \in D(\varphi)$, we can put $y = \sum_{k=0}^\infty \ip{x}{\varphi_k}e_k$. Then
$$ \ip{x}{Te_j} = \ip{x}{\varphi_j}= \ip{y}{e_j}.$$
This equality extend obviously to $D_e$; then $x \in D(T_0^*)$ and $T_0^* x= y = \sum_{k=0}^\infty \ip{x}{\varphi_k}e_k$.
The equality $D(\varphi)= D(T_0^*)$ in turn implies that $D(\varphi)$ is dense in $\Hil$, since $D(T_0^*)$ is dense.
\end{proof}
Thus, if $(\F_e, T)$ and $(\F_{e'}, T')$ are both constructing pairs for $\F_\varphi$ one has $D(T_0^*)=D({T'}_0^*)$ also in the case when $\F_e$ and $\F_{e'}$ are different.

Incidentally, the previous argument shows that $D(\varphi)$ is a Hilbert space under the norm
$$\|x\|_\varphi =\left( \|x\|^2 + \sum_{k=0}^\infty |\ip{x}{\varphi_k}|^2\right)^{1/2}.$$

Moreover, if $x \in \Hil$ and  $\ip{x}{\varphi_k}=0$, for every $k \in {\mathbb N}$ (or, equivalently $\sum_{k=0}^\infty |\ip{x}{\varphi_k}|^2=0$); then $\ip{x}{T_0e_k}=0$, for every $k \in {\mathbb N}$.
This implies that $x \in D(T_0^*)$ and $\ip{T_0^*x}{e_k}=0$, for every $k \in {\mathbb N}$; hence $T_0^*x=0$.

The best situation occurs if $D_e$ is a core for $T$; i.e., if $T=\overline{T_0}$. In this case, of course $D(\varphi)=D(T^*)$ and if $\sum_{k=0}^\infty |\ip{x}{\varphi_k}|^2=0$ then $x=0$; that is the sequence $\{\varphi_n\}$ is total in $H$; i.e. $D_\varphi$ is dense in $\Hil$.

\begin{prop} If $D_e$ is a core for $T$, the linear span $D_\varphi$ of $\{\varphi_n\}$ is dense in $\Hil$.
Conversely, if $D_\varphi$ is dense in $\Hil$ and $T$ is bounded from below (i.e., if $T^{-1}$ is bounded), $D_e$ is a core for $T$.
\end{prop}
\begin{proof} Taking into account the previous discussion, we need only to prove the second statement. As is well-known, $D_e$ is a core for $T$ if and only if $D_e$ is dense in $D(T)$ considered as a Hilbert space with the graph norm $\|z\|_T= (\|z\|^2 +\|Tz\|^2)^{1/2}$, $z\in D(T)$; since $T$ is bounded from below, $\|z\|_T$ can be replaced by $\|T  z\|$. Let $y\in D(T)$ and suppose that $\ip{Ty}{Te_k}=0$, for every $k \in {\mathbb N}$.
Then, $\ip{Ty}{\varphi_k}=0$, for every $k \in {\mathbb N}$. This implies that $Ty=0$ and in turn $y=0$.
\end{proof} }

\begin{prop} Let $\{ \varphi_n \}$ be a generalized Riesz system, with constructing pair $(\F_e, T)$. Suppose that, for every $n \in {\mathbb N}$, $e_n\in D((T^{-1})^*)$ and define $\psi_n =(T^{-1})^*e_n$. Then the sequences $\{ \varphi_n \}$ and $\{ \psi_n \}$ are biorthogonal. Moreover, $\{ \psi_n \}$ is a generalized Riesz system, with constructing pair $(\F_e, (T^{-1})^*)$.
\end{prop}
\begin{proof}
The statement about biorthogonality is straightforward. Let us show that $\{ \psi_n \}$ is a generalized Riesz system. The operator $(T^{-1})^*$ is closed and densely defined, being the adjoint of a closed densely defined operator. Its inverse $T^*$ is also closed and densely defined. Therefore $\{ \psi_n \}$ is a generalized Riesz system, with constructing pair $(\F_e, (T^{-1})^*)$.
\end{proof}

On the other hand, let us suppose that $\{ \varphi_n \}$ and $\{ \psi_n \}$ are biorthogonal Riesz systems with constructing pairs $(\F_e, T)$, $(\F_{e'}, S)$, respectively. As before if $X$ is an operator defined on $D_e$ we denote by $X_0$ the restriction of $X$ to $D_e$. Since $\F_e$ and $\F_{e'}$ are ONB in $\Hil$, there exists a unitary operator such that $e'_n=Ue_n$, for every $n \in {\mathbb N}$, and by the biorthogonality of $\{ \varphi_n \}$ and $\{ \psi_n \}$, we obtain
$$\ip{Te_n}{SUe_m}=\ip{Te_n}{Se'_m}=\ip{\varphi_n}{\psi_m}=\delta_{n,m}.$$
These equalities imply that $SUe_m \in D(T_0^*)$ and $T_0^*SUe_m=e_m$, , for every $m \in {\mathbb N}$. Analogously, $Te_m \in D((SU)_0^*)$ and $(SU)_0^*Te_m=e_m$, , for every $m \in {\mathbb N}$. Hence $T_0$ is invertible and $T_0^{-1}\subseteq (SU)_0^*.$

\section{Main Theorem}

 Let $\F_\varphi=\{\varphi_n, n\geq0\}$ and  $\F_\psi=\{\psi_n, n\geq0\}$ be two sets of vectors of $\Hil$. Let us now assume that the sets $D(\varphi)$ and  $D(\psi)$ {defined as in \eqref{eqn_one}} are such that their intersection
 $D(\varphi)\cap D(\psi)$ is dense in $\Hil$. Of course, this implies that $D(\varphi)$ and $D(\psi)$ are dense in $\Hil$, too.

It is not hard to imagine a concrete example in which this happens, as the following example shows.

\vspace{2mm}

{\bf Example 1:--} Let $\E=\{e_n(x)=\frac{1}{\sqrt{2^n n! \sqrt{\pi}}} H_n(x)\,e^{-x^2/2}\}$ be the orthonormal basis of the eigenvector of the quantum harmonic oscillator. Here  $H_n(x)$ is the n-th Hermite polynomial.

Let now $X$  be the following multiplication operator: $(Xf)(x)=(1+x^2)f(x)$, for each $f$ in its domain $D(X)=\{f\in L^2(\Bbb R):\, (1+x^2)f(x)\in L^2(\Bbb R)\}$. Since $D(X)\supseteq \Sc(\R)$, the Schwartz test functions space, $D(X)$ is dense in $L^2(\R)$. It is clear that $X$ is not everywhere defined and that it admits a bounded inverse. Moreover, for each $n\in {\mathbb N}$, $e_n(x)\in D(X)$. Then, if we define the functions $\varphi_n(x)=(Xe_n)(x)$ and $\psi_n(x)=(X^{-1}e_n)(x)$, $ n\in {\mathbb N}$, we can easily see that $D(\varphi) \cap D(\psi)=D(X)$.

It is clear that both $\{\varphi_n\}$ and $\{\psi_n\}$ are generalized Riesz system in the sense of Definition \ref{def11}. Hence, by Proposition \ref{addprop1}, $D(\varphi)=D(X_0^*)$, where $X_0$ denotes the restriction of $X$ to the linear span $\D_e$ of $\{e_n\}$, while $D(\psi)=L^2(\Bbb R)$. Moreover, $\{\varphi_n\}$ and $\{\psi_n\}$ are biorthogonal:
$$
\ip{\varphi_n}{\psi_k}=\ip{Xe_n}{X^{-1}e_k}=\ip{e_n}{e_k}=\delta_{n,k}.
$$

\vspace{2mm}

We define, as in the Introduction, the following positive sesquilinear forms $\Omega_\varphi$ and $\Omega_\psi$:
\be
\left\{
    \begin{array}{ll}
\Omega_\varphi(x,y)=\sum_{n=0}^\infty\ip{x}{\varphi_n}\ip{\varphi_n}{y} & {\rm on} \;\; D(\varphi) \times D(\varphi)\\
\Omega_\psi(x,y)=\sum_{n=0}^\infty\ip{x}{\psi_n}\ip{\psi_n}{y} & {\rm on} \;\; D(\psi) \times D(\psi).\\
       \end{array}
        \right.
\label{21}\en
{It is clear that} they are well defined. In fact, for instance,
$$
|\Omega_\varphi(x,y)|\leq \sqrt{\sum_{n=0}^\infty|\ip{x}{\varphi_n}|^2}\,\sqrt{\sum_{n=0}^\infty|\ip{\varphi_n}{y}|^2}<\infty,
$$
in view of the definition of $D(\varphi)$. 
 In particular,  $D(\varphi)$ (respectively, $D(\psi)$)  is the largest subspace where $\Omega_\varphi$ (respectively,  $\Omega_\psi$) can be defined. Furthermore, the positive sesquilinear forms $\Omega_\varphi$ and $\Omega_\psi$ are closed. Indeed, it is easily shown that $D(\varphi)$ is a Hilbert space with inner product:
\begin{eqnarray}
\ip{x}{y}_\varphi \equiv \Omega_\varphi(x,y) + \ip{x}{y}, \quad x,y\in D(\varphi). \nonumber
\end{eqnarray}
Similarly, $D(\psi)$ is  a Hilbert space with scalar product
\begin{eqnarray}
\ip{x}{y}_{\psi} \equiv \Omega_\psi(x,y) + \ip{x}{y},  \quad x,y\in D(\psi).\nonumber
\end{eqnarray}

As already stated, by the representation theorem for sesquilinear forms, \cite{kato, davies}, there exist uniquely determined positive and self-adjoint operators, $K_\varphi$ and $K_\psi$, with $D(K_\varphi^{1/2})=D(\varphi)$ and $D(K_\psi^{1/2})=D(\psi)$, such that
\be
\Omega_\varphi(x,y)=\ip{K_\varphi^{1/2}x}{K_\varphi^{1/2}y},\qquad \Omega_\psi(x',y')=\ip{K_\psi^{1/2}x'}{K_\psi^{1/2}y'},
\label{22}\en
for all $x,y\in D(\varphi)$ and $x',y'\in D(\psi)$.

Suppose that the sets $\F_\varphi$ and $\F_\psi$ are biorthogonal, that is, $\ip{\varphi_k}{\psi_l}= \delta_{kl}$ for $k, \; l\in \Bbb N$. Let us call $D_\varphi := {\rm linear \; span \; of}$ $\F_\varphi$ and $D_\psi := {\rm linear \; span \; of}$ $\F_\psi$. Then $D_\psi \subseteq D(\varphi)$ and $ D_\varphi \subseteq D( \psi)$. Hence, by the above formulas (\ref{21}) and (\ref{22}), we see that
$$
\Omega_\varphi(\psi_k,y)=\ip{\varphi_k}{y}=\ip{K_\varphi^{1/2}\psi_k}{K_\varphi^{1/2}y}, \quad \forall y\in D(\varphi).
$$
Hence, $\psi_k\in D(K_\varphi)$ and $\varphi_k=K_\varphi\psi_k$. Analogously we can prove that $\varphi_k\in D(K_\psi)$ and that $\psi_k=K_\psi\varphi_k$, so that
\be
\psi_k=K_\psi K_\varphi\psi_k,\qquad  \varphi_k=K_\varphi K_\psi\varphi_k, \quad \forall k\in {\mathbb N}.
\label{23}\en
 Of course these equalities extend to $D_\psi$ and $D_\varphi$, respectively.

 However, we observe that   it is not true, in general, that
$$
K_\varphi K_\psi f=f, \qquad \forall f\in D(K_\varphi K_\psi)=\{h\in D(K_\psi):\, K_\psi h\in D(K_\varphi)\},
$$
and
$$
K_\psi K_\varphi g=g, \qquad \forall g\in D(K_\psi K_\varphi)=\{h\in D(K_\varphi):\, K_\varphi h\in D(K_\psi)\}.
$$
In fact, even if $D_\varphi$ and $D_\psi$ are dense in $\Hil$, in order to extend (\ref{23}) to $D(K_\psi K_\varphi)$ and $D(K_\varphi K_\psi)$, more conditions are needed. For instance, we could require that $D_\psi$ is a core for $\overline{K_\psi K_\varphi}$.

 \vspace{2mm}

{\bf Example 1, part 2:--} Let us come back to the situation described in the previous part of the Example 1. In this case, as seen before, $D(\varphi)=D(X_0^*)$. Then, for $f,g\in D(\varphi)$, we have
$$
\Omega_\varphi (f,g)
= \sum_{k=0}^\infty \ip{f}{\varphi_k}\ip{\varphi_k}{g}=\sum_{k=0}^\infty \ip{f}{Xe_k}\ip{X e_k}{g}=\sum_{k=0}^\infty \ip{X_0^*f}{e_k}\ip{ e_k}{X_0^*g}=\ip{X_0^*f}{X_0^*g}.
$$
Similarly,
$$
\Omega_\psi (f,g)
= \ip{X^{-1}f}{X^{-1}g},
$$
for all $f,g\in L^2(\R)$. In this case $K_\varphi=\overline{X_0} X_0^*$ and $K_\psi=X^{-2}$.

\vspace{2mm}

A useful working assumption on $\F_\varphi$ and $\F_\psi$, often satisfied in concrete physical models \cite{bagbook} and which we systematically adopt here,  is that they are $\D$-quasi bases, i.e. that for all $x,y\in\D$ the following identities hold:
\be
\ip{x}{y}=\sum_{n=0}^\infty\ip{x}{\varphi_n}\ip{\psi_n}{y}=\sum_{n=0}^\infty\ip{x}{\psi_n}\ip{\varphi_n}{y},
\label{25}\en
where $\D$ is a dense subspace in $\Hil$.
We put
$$
 \Omega_{\varphi,\psi}(x,y)=\sum_{n=0}^\infty\ip{x}{\varphi_n}\ip{\psi_n}{y},
$$
and
$$
 \Omega_{\psi,\varphi}(x,y)=\sum_{n=0}^\infty\ip{x}{\psi_n}\ip{\varphi_n}{y}
$$
for all $x,y$ for which these make sense. {This form is in general neither semi-bounded nor sectorial. Thus, Kato's representation theorems cannot be applied. This would be not a major problem since several variants to these famous theorems have been proposed (we refer to \cite{corso_ct,dibella_ct}, where the notion of {\em solvable form} has been introduced and studied, and for a rather complete bibliography on this matter). However, all this is of little use for us since
(\ref{25}) implies that
 $\Omega_{\varphi,\psi}$ and $\Omega_{\psi,\varphi}$ are both positive. This also explains why this possibility was not excluded from the very beginning.}

We now investigate when $\{ \varphi_n \}$ and $\{ \psi_n \}$ are generalized Riesz systems making use of the above operators $K_\varphi^{1/2}$ and $K_\psi^{1/2}$. For that, the notion of $\D$-quasi bases will be relevant. We get the following main theorem.
\vspace{2mm}

\begin{thm}\label{thm1}
 Let $\D$ be a dense subspace in $\Hil$ such that
\be
\F_\varphi \cup \F_\psi \subseteq \D \subseteq D(\varphi) \cap D(\psi),
\en
and denote by
\begin{eqnarray}
R_\varphi &:=& K_\varphi^{1/2}\lceil_ \D \; ({\rm the \; restriction \; of }\; K_\varphi^{1/2} \; {\rm to} \; \D) \nonumber \\
 \mbox{and} \;\;\; R_\psi &:=& K_\psi^{1/2}\lceil_ \D \; ({\rm the \; restriction \; of }\; K_\psi^{1/2} \; {\rm to} \; \D).
\end{eqnarray}
Then the following statements are equivalent.
\par
\hspace{3mm} (i)  (i)$_1$ $\F_\varphi $ and $\F_\psi $ are biorthogonal and $\D$-quasi bases;
\par
\hspace{9mm} (i)$_2$ there exist dense subspaces $\E$ and $\E'$ in $\Hil$ such that $\F_\psi \subseteq \E \subseteq \D$, $K_\varphi^{1/2}\E \subseteq \D$
\par
\hspace{9mm} and $\F_\varphi \subseteq \E' \subseteq \D$, $K_\psi^{1/2}\E' \subseteq \D$;
\par
\hspace{9mm} (i)$_3$ {$R_\varphi^{\ast}K_\psi^{1/2} =\1$ on $\D$, i.e. $R_\varphi^{\ast}K^{1/2}\varphi =\varphi$, for all $\varphi\in\D$}.
\par
\hspace{3mm} (ii) $\F_\varphi$ is a generalized Riesz system with a constructing pair $\left( \overline{R_\varphi}, \; \F_e \right)$ and $\F_\psi$
\par
\hspace{3mm} is a generalized Riesz system with a constructing pair $\left( \overline{R_\psi}, \; \F_e  \right)$, where $\F_e$ is an
\par
\hspace{3mm} orthonormal basis in $\Hil$ contained in $\D \cap D(K_\varphi)\cap D(K_\psi)$.
\end{thm}
\begin{proof} (i)$\Rightarrow$(ii) We put
\be
e_{n} =K_\varphi^{1/2} \psi_{n} \;\; {\rm and } \;\; e'_{n} =K_\psi^{1/2} \varphi_{n}, \;\; n=0,1, \ldots . \nonumber
\en
By assumptions (3.5) and (i)$_2$, we have
\be
e_{n}, \; e'_{n} \in \D, \;\; n=0,1, \ldots .
\en
By (\ref{22}), we have
$$
\ip{ e_{k} }{ K_\varphi^{1/2} y}
= \ip{K_\varphi^{1/2}\psi_{k} }{K_\varphi^{1/2}y}
= \Omega_\varphi(\psi_{k},y)
= \sum_{n=0}^{\infty} \ip{\psi_{k} }{\varphi_{n}}\ip{\varphi_{n} }{y}
= \ip{\varphi_{k} }{y}
$$
for all $y \in D(\varphi)=D(K_\varphi^{1/2})$. Hence, it follows that
\be
e_{k} \in D(K_\varphi^{1/2}) \;\; {\rm and} \;\; K_\varphi^{1/2} e_{k}= \varphi_{k}, \;\; k= 0,1 \ldots .
\en
Hence, we have
\be
K^{1/2} e_k=\varphi_k \in \D \subseteq D(K_\varphi^{1/2}), \nonumber
\en
which implies
\be
e_k\in D(K_\varphi), \;\; n=0,1, \ldots .
\en
Since $\F_\varphi$ and $\F_\psi$ are biorthogonal, it follows that
$$
\ip{e_{m} }{e_{k}}
= \ip{K_\varphi^{1/2}\psi_{m} }{K_\varphi^{1/2}\psi_{k}}
= \Omega_\varphi (\psi_{m},\psi_{k})
= \sum_{n=0}^{\infty} \ip{\psi_{m} }{\varphi_{n}}\ip{\varphi_{n}}{\psi_{k}}
= \delta_{mk},
$$
which means that $\{ e_{k} \}$ is an orthonormal system in $\Hil$. In a similar way, we get that $\{ e'_{k} \}$ is an orthonormal system in $\Hil$ and
\be
\psi_k =K_\psi^{1/2}e'_k \;\; {\rm and} \;\; e'_k \in D(K_\psi), \;\; k=0,1, \ldots .
\en
By (3.5), (3.10) and (i)$_3$ we have
\be
e_n = K_\varphi^{1/2}\psi_n=K_\varphi^{1/2}K_\psi^{1/2}e'_n=e'_n , \;\; n=0,1, \ldots .\nonumber
\en
Hence, by (3.7), (3.9) and (3.10) we have
\be
\D_e \subseteq \D\cap D(K_\varphi) \cap D(K_\psi),
\en
where $\D_e$ is the linear span of the $e_n$'s.
Furthermore, since $\F_\varphi$ and $\F_\psi$ are $\D$-quasi bases in $\Hil$ and because of the assumption (i)$_2$: $K_\varphi^{1/2} \E \subseteq \D$ and $K_\psi^{1/2} \E' \subseteq \D$, it follows from (i)$_3$ that
\begin{eqnarray}\label{211}
\sum_{n=0}^{\infty} \ip{x }{e_n}\ip{e_n }{y}
&=& \sum_{n=0}^{\infty} \ip{K_\varphi^{1/2}x }{\psi_n}\ip{\varphi_n }{K_\psi^{1/2}y} \nonumber \\
&=& \ip{K_\varphi^{1/2}x }{K_\psi^{1/2}y}
= \ip{x }{K_\varphi^{1/2}K_\psi^{1/2}y}
= \ip{x }{y}
\end{eqnarray}
for all $x\in \E$ and $y \in \E'$. Due to the orthogonality of $\F_e$, (\ref{211}) can be extended to all of $\Hil$. Indeed we have, taking $x\in \Hil$ and $\{x_k\}$ a sequence of elements of $\E$ converging to $x$,
$$
\left|\sum_{h=n+1}^m \ip{x_k }{e_h}\ip{e_h }{y}-\sum_{h=n+1}^m \ip{x }{e_h}\ip{e_h }{y}\right|= \left|\sum_{h=n+1}^m \ip{x_k-x }{e_h}\ip{e_h }{y}\right|\leq
$$
$$
\leq \left(\sum_{h=n+1}^m |\ip{x_k-x }{e_h}|^2\right)^{1/2}\left(\sum_{h=n+1}^m |\ip{e_h }{y}|^2\right)^{1/2}\leq \|x_k-x\|\|y\|\rightarrow 0
$$
when $k$ diverges, for all $y\in\Hil$.

 Hence $\F_e$ is an orthonormal basis in $\Hil$. 

{Let us consider again the operators $R_\varphi$, $R_\psi$ of (3.6). For short, we put ${\sf R}_\varphi=\overline{R_\varphi}$ and ${\sf R}_\psi=\overline{R_\psi}$.  Then ${\sf R}_\varphi$ is a densely defined closed operator in $\Hil$ with densely defined inverse such that
\be
{\sf R}_\varphi e_n =\varphi_n \;\; {\rm and } \;\; ({\sf R}_\varphi^{-1})^{\ast} e_n=K_\psi^{1/2}e_n=\psi_n, \;\; n=0,1, \ldots . \nonumber
\en}
Hence $\F_\varphi$ is a generalized Riesz system with a constructing pair $({\sf R}_\varphi, \F_e)$. Similarly $\F_\psi$ is a generalized Riesz system with a constructing pair $({\sf R}_\psi, \F_e)$. \\
(ii)$\Rightarrow$(i) By the assumption (ii), ${\sf R}_\varphi$ has a densely defined inverse and $ \left( {\sf R}_\varphi^{-1} \right)^{\ast}e_n = \psi_n$, $n=0,1, \ldots$. Hence we have
$$
\ip{\varphi_k}{\psi_l}
= \ip{K_\varphi^{1/2}e_k}{ {\sf R}_\varphi^{-1}e_l}
= \ip{e_k}{e_l}
= \delta_{kl}
$$
and
\be
\Omega_{\varphi,\psi}(x,y)=\sum_{n=0}^{\infty} \ip{x}{\varphi_n}\ip{\psi_n}{y}
= \sum_{n=0}^{\infty} \ip{K_\varphi^{1/2}x}{e_n}\ip{e_n}{ {\sf R}_\varphi^{-1}y}
= \ip{K_\varphi^{1/2}x}{ {\sf R}_\varphi^{-1}y}
= \ip{x}{y}
\label{add1}\en
for all $x,y \in \D$. Thus (i)$_1$ holds. We next show (i)$_2$. Since $K_\varphi^{1/2}e_n=\varphi_n$ and $K_\psi^{1/2}e_n=\psi_n$, $n=0,1, \ldots$, we have
$$
\ip{K_\varphi^{1/2} \psi_n}{ e_k}
= \ip{\psi_n}{\varphi_k}
= \delta_{nk}
= \ip{e_n}{e_k}
= \ip{\varphi_n}{ \psi_k}
= \ip{K_\psi^{1/2}\varphi_n}{e_k}
$$
for all $k$. Hence
\be
K_\varphi^{1/2} \psi_n=K_\psi^{1/2} \varphi_n =e_n , \;\; n=0,1, \ldots .
\en
We denote by $\E$ the subspace of $\D$ generated by $\{ x \in \D ; K_\varphi^{1/2} x \in \D \}$ and denote by $\E'$ the subspace generated by $\{ y \in \D; K_\psi^{1/2} y\in \D \}$.
Then it follows from (2.12) that $\F_\psi \cup \{ e_n \} \subseteq \E \subseteq \D$ and $\F_\varphi \cup \{ e_n \} \subseteq \E' \subseteq \D$. Hence $\E$ and $\E'$ are dense subspaces in $\Hil$. It is clear that $K_\varphi^{1/2} \E \subseteq \D$ and $K_\psi^{1/2} \E' \subseteq \D$. Thus (i)$_2$ holds. Finally we show (i)$_3$. Indeed, it follows from (3.12) that
\be
\ip{e_n}{{\sf R}_\psi^{-1} K_\psi^{1/2} y} = \ip{e_n}{y} \nonumber
\en
and
$$
\ip{e_n}{ {\sf R}_\varphi^{\ast} K_\psi^{1/2}y }
= \ip{\varphi_n}{K_\psi^{1/2}y}
= \ip{K_\psi^{1/2} \varphi_n}{y}
= \ip{e_n}{y}
$$
for all $n$ and $y \in \D$, which implies that ${\sf R}_\varphi^{\ast} K_\psi^{1/2}=\1$ on $\D$. This completes the proof.

\end{proof}

In particular, equation (\ref{add1}) shows that, under our assumptions, $\Omega_{\varphi,\psi}$ is, in fact, positive defined. The same conclusion can be deduced for $\Omega_{\psi,\varphi}$, with similar arguments. In case that $\D := D(\varphi)\cap D(\psi)=D(K_\varphi^{1/2}) \cap D(K_\psi^{1/2})$, we have the following\\
\par
\begin{cor}\label{cor1} Let $\D =D(\varphi) \cap D(\psi)$. Then the following statements are equivalent.
\par
\hspace{3mm} (i)  (i)$_1$ $\F_\varphi $ and $\F_\psi $ are biorthogonal and $\D$-quasi bases;
\par
\hspace{9mm} (i)$_2$ $\F_\varphi \cup \F_\psi \subseteq \D$ and $D(K_\varphi) \cap D(K_\psi)$ is dense in $\Hil$;
\par
\hspace{9mm} (i)$_3$ ${\sf R}_\varphi^{\ast}K_\psi^{1/2} =\1$ on $\D$.
\par
\hspace{3mm} (ii) $\F_\varphi$ is a generalized Riesz system with  constructing pair $\left( {\sf R}_\varphi, \; \F_e \right)$ and $\F_\psi$
\par
\hspace{3mm} is a generalized Riesz system with  constructing pair $\left( {\sf R}_\psi, \; \F_e  \right)$, where $\F_e$ is an
\par
\hspace{3mm} orthonormal basis in $\Hil$ contained in $ D(K_\phi)\cap D(K_\psi)$.\\
\end{cor}
\begin{proof} By (i)$_2$ the assumption (3.5) in Theorem 3.1 holds and we put $\E := D(K_\varphi)\cap D(K_\psi^{1/2})$ and $\E' := D(K_\varphi^{1/2}) \cap D(K_\psi)$ satisfy (i)$_2$ in Theorem 3.1. By Theorem 3.1, the implication (i) $\Rightarrow$ (ii) holds.\\
(ii) $\Rightarrow$ (i) Since $\varphi_n=K_\varphi^{1/2} e_n$, $\psi_n=K_\psi^{1/2}e_n$ and $\D_e \subseteq D(K_\varphi)\cap D(K_\psi)$,
we have $\F_\varphi \cup \F_\psi \subseteq \D$. Furthermore, since $\{e_n\} \subseteq D(K_\varphi) \cap D(K_\psi)$, $D(K_\varphi) \cap D(K_\psi)$ is dense in $\Hil$. Thus (i)$_2$ holds. By the proof (ii) $\Rightarrow$ (i) in Theorem 3.1, (i)$_1$ and (i)$_3$ hold. This completes the proof.

\end{proof}

Corollary \ref{cor1} is, in a sense, more intrinsic than Theorem \ref{thm1}, since it does not involve {\em external objects} such as the subspaces $\E$ and $\E'$, which on the other hand are needed in Theorem \ref{thm1}, and it is more useful applications as those we will consider in the next Section.

\section{Connections with some Hamiltonians}

In this Section, we will consider how to construct well-defined physical operators from a biorthogonal pair $(\F_\varphi ,\F_\psi)$. The assumption that $\F_\varphi$ and $\F_\psi$ are regular (that is, $D_\varphi$ and $D_\psi$ are dense in $\Hil$) is useful to define Hamiltonian like operators as in \cite{hiro1, hiro2} which are densely defined. However, we don't assume in general that $D_\varphi$ and $D_\psi$ are dense in $\Hil$.

Let ${\balpha}:=\{\alpha_n \}$ be a sequence of
complex numbers, $\F_\varphi$ and $\F_\psi$ biorthogonal $\D$-quasi bases, and $H_{\varphi, \psi}^\balpha  $ and $H_{\psi,\varphi}^\balpha  $ two operators defined
as follows:
$$D(H_{\varphi, \psi}^\balpha  )=\left\{x \in\Hil:\,  \sum_{n=1}^\infty \alpha_n\left<x,\psi_n\right>\,\varphi_n \mbox{ exists in }\Hil\right\},$$
 $$D(H_{\psi,\varphi}^\balpha  )=\left\{y \in\Hil:\, \sum_{n=1}^\infty \alpha_n\left<y,\varphi_n\right>\,\psi_n\mbox{ exists in }\Hil\right\},$$
and
\begin{eqnarray}
H_{\varphi, \psi}^\balpha   x &:=& \left( \sum_{n=0}^\infty \alpha_n \varphi_n \otimes \bar{\psi}_n \right) x :=\sum_{n=1}^\infty \alpha_n\left<x,\psi_n\right>\,\varphi_n,  \nonumber \\
H_{\psi,\varphi}^\balpha   y &:=& \left( \sum_{n=0}^\infty \alpha_n \psi_n \otimes \bar{\varphi}_n \right) y :=\sum_{n=1}^\infty \alpha_n\left<y ,\varphi_n\right>\,\psi_n,
\end{eqnarray}
for all $x\in D(H_{\varphi, \psi}^\balpha  )$ and $y \in D(H_{\psi,\varphi}^\balpha)$. It is clear that $\D_\psi \subseteq D(H_{\psi,\varphi}^\balpha  )$ and  $\D_\varphi  \subseteq D(H_{\varphi, \psi}^\balpha  )$, and that
$$
H_{\varphi, \psi}^\balpha   \varphi_k =\alpha_k \varphi_k,  \; \hspace{2.5cm}
 H_{\psi,\varphi}^\balpha   \psi_k =\alpha_k \psi_k, $$
 for all $k$.
 Therefore, the $\varphi_n$'s
and the $\psi_n$'s are eigenstates respectively of $ H_{\varphi, \psi}^\balpha$
and $ H_{\psi,\varphi}^\balpha$, and the complex numbers $\alpha_n$'s are
their (common) eigenvalues.

If $\F_\varphi$ and $\F_\psi$ are regular, then the operators $H_{\varphi,\psi}^{\bm{\alpha}}$ and $H_{\psi,\varphi}^{\bm{\alpha}}$ are densely defined. Furthermore, $(H_{\psi,\varphi}^{\bm{\alpha}})^\ast \supseteq H_{\varphi,\psi}^{\overline{\bm{\alpha}}}$, where $\overline{\bm{\alpha}} =\{ \bar{\alpha}_n \}$, and they coincide if $\F_\varphi$ and $\F_\psi$ are Riesz bases.
In \cite{bb2017} it is also considered the possibility of factorizing these operators in terms of suitably defined ladder operators. We will not repeat the same analysis here, while we focus on other relevant operators considered in \cite{bb2017}, to relate them to some of the operators introduced here out of the sesquilinear forms we have considered.
But, it is difficult to investigate concretely these operators because  $\{ \varphi_n \}$ and $\{ \psi_n \}$ are not orthogonal systems in $\Hil$. For example, if $D_\varphi$ and $D_\psi$ are not dense in $\Hil$, we don't know whether these operators are densely defined or not. If $\F_\varphi$ is a generalized Riesz system with a constructing pair $(\F_{\bm{e}},T)$ and $\psi_n := (T^\ast)^{-1} e_n$, $n=0,1, \ldots$, then $\F_\varphi$ and $\F_\psi$ are biorthogonal and we can define the following  non-self-adjoint Hamiltonians:
\begin{eqnarray}
T \left( \sum_{n=0}^\infty \alpha_n e_n \otimes \bar{e}_n \right) T^{-1}, \\
(T^\ast)^{-1} \left( \sum_{n=0}^\infty \alpha_n e_n \otimes \bar{e}_n \right) T^\ast.
\end{eqnarray}
Here we denote the operator $\sum_{n=0}^\infty \alpha_n e_n \otimes \bar{e}_n$ by $H_{\bm{e}}^{\bm{\alpha}}$. Since $\{ e_n \}$ is an ONB in $\Hil$, it is easily shown that $D(H_{\bm{e}}^{\bm{\alpha}}) = \{ x\in \Hil ; \sum_{n=0}^{\infty} |\alpha_n|^2 |<x,e_n>|^2 < \infty \}$ and $H_{\bm{e}}^{\bm{\alpha}}$ is a densely defined closed operators in $\Hil$ satisfying $(H_{\bm{e}}^{\bm{\alpha}})^\ast =H_{\bm{e}}^{\overline{\bm{\alpha}}}$. Hence, if $\bm{\alpha} = \{ \alpha_n \} \subset \R$, then $H_{\bm{e}}^{\bm{\alpha}}$ is a self-adjoint operator and can be understood as a standard self-adjoint Hamiltonian. If $\F_\varphi$ is a Riesz basis, then $H_{\varphi,\psi}^{\bm{\alpha}} =TH_{\bm{e}}^{\bm{\alpha}}T^{-1}$ and $H_{\psi,\varphi}^{\bm{\alpha}}= (T^\ast)^{-1}H_{\bm{e}}^{\bm{\alpha}}T^\ast$, but these operators don't coincide in general. It is easier to investigate the operators $TH_{\bm{e}}^{\bm{\alpha}}T^{-1}$ and $(T^\ast)^{-1}H_{\bm{e}}^{\bm{\alpha}}T^\ast$ than to work directly with the operators in (4.1) $H_{\varphi,\psi}^{\bm{\alpha}}= \sum_{n=0}^\infty \alpha_n \varphi_n \otimes \bar{\psi}_n$ and $H_{\psi,\varphi}^{\bm{\alpha}}=\sum_{n=0}^\infty \alpha_n \psi_n \otimes \bar{\varphi}_n$. Hence, when $\F_\varphi$ is a generalized Riesz system with a constructing pair $(\F_{\bm{e}},T)$, the operators $TH_{\bm{e}}^{\bm{\alpha}}T^{-1}$ and $(T^\ast)^{-1}H_{\bm{e}}^{\bm{\alpha}}T^\ast$ can be regarded as non-self-adjoint Hamiltonians. For convenience we will often denote them by $H_{\varphi,\psi}^{\bm{\alpha}}$ and $H_{\psi,\varphi}^{\bm{\alpha}}$:
\begin{eqnarray}
H_{\varphi,\psi}^{\bm{\alpha}} =TH_{\bm{e}}^{\bm{\alpha}}T^{-1} \;\;\; {\rm and} \;\;\; H_{\psi,\varphi}^{\bm{\alpha}} =(T^\ast)^{-1}H_{\bm{e}}^{\bm{\alpha}}T^\ast.
\end{eqnarray}
Next we define the generalized lowering and raising operators for $\F_\varphi$ and $\F_\psi$. We put first
\begin{eqnarray}
A_{\bm{e}}^{\bm{\alpha}}
&:=& \sum_{n=0}^\infty \alpha_{n+1} e_n \otimes \bar{e}_{n+1}, \nonumber \\
B_{\bm{e}}^{\bm{\alpha}}
&:=& \sum_{n=0}^\infty \alpha_{n+1} e_{n+1} \otimes \bar{e}_{n}. \nonumber
\end{eqnarray}
Then, $B_{\bm{e}}^{\bm{\alpha}}= (A_{\bm{e}}^{\overline{\bm{\alpha}}})^\ast$ and
\begin{eqnarray}
A_{\bm{e}}^{\bm{\alpha}} e_n
&=& \left\{
\begin{array}{ccc}
&0,& \;\;\; n=0 \\
&\alpha_n e_{n-1},& \quad\qquad n=1,2, \ldots \\
\end{array}
\right. \nonumber \\
B_{\bm{e}}^{\bm{\alpha}} e_n
&=& \alpha_{n+1} e_{n+1}, \qquad\quad\qquad n=0,1, \ldots  \nonumber
\end{eqnarray}
and so they are called the lowering and raising operators for $\{ e_n \}$, respectively. We now define the following operators:
\begin{eqnarray}
\left\{
\begin{array}{ccc}
A_{\varphi,\psi}^{\bm{\alpha}}
&:=& TA_{\bm{e}}^{\bm{\alpha}}T^{-1} = T \left( \sum_{n=0}^\infty \alpha_{n+1} e_n \otimes \bar{e}_{n+1} \right) T^{-1}, \\
B_{\varphi,\psi}^{\bm{\alpha}}
&:=& TB_{\bm{e}}^{\bm{\alpha}}T^{-1} = T \left( \sum_{n=0}^\infty \alpha_{n+1} e_{n+1} \otimes \bar{e}_{n} \right) T^{-1}, \\
\end{array}
\right. \nonumber
\end{eqnarray}
\begin{eqnarray}
\left\{
\begin{array}{ccc}
A_{\psi,\varphi}^{\bm{\alpha}}
&:=& (T^\ast)^{-1}A_{\bm{e}}^{\bm{\alpha}}T^\ast =(T^\ast)^{-1} \left( \sum_{n=0}^\infty \alpha_{n+1} e_n \otimes \bar{e}_{n+1} \right) T^\ast, \\
B_{\varphi,\psi}^{\bm{\alpha}}
&:=& (T^\ast)^{-1}B_{\bm{e}}^{\bm{\alpha}}T^\ast = (T^\ast)^{-1} \left( \sum_{n=0}^\infty \alpha_{n+1} e_{n+1} \otimes \bar{e}_{n} \right) T^\ast  . \\
\end{array}
\right. \nonumber
\end{eqnarray}
Then the following results are easily shown.\\
\par
\begin{lemma} (1) $D(A_{\varphi,\psi}^{\bm{\alpha}}) \cap D(B_{\varphi,\psi}^{\bm{\alpha}}) \supseteq D_\varphi$ and
\begin{eqnarray}
A_{\varphi,\psi}^{\bm{\alpha}} \varphi_n
&=& \left\{
\begin{array}{ccc}
&0 &\;\;\; ,n=0 \\
&\alpha_n \varphi_{n-1}& \;\;\; ,n=1,2, \ldots \\
\end{array}
\right. \nonumber \\
B_{\varphi,\psi}^{\bm{\alpha}} \varphi_n
&=& \alpha_{n+1} \varphi_{n+1}, \;\;\; n=0,1, \ldots . \nonumber
\end{eqnarray}
(2) $D(A_{\psi,\varphi}^{\bm{\alpha}}) \cap D(B_{\psi,\varphi}^{\bm{\alpha}}) \supseteq D_\psi$ and
\begin{eqnarray}
A_{\psi,\varphi}^{\bm{\alpha}} \psi_n
&=& \left\{
\begin{array}{ccc}
&0 &\;\;\; ,n=0 \\
&\alpha_n \psi_{n-1}& \;\;\; ,n=1,2, \ldots \\
\end{array}
\right. \nonumber \\
B_{\psi,\varphi}^{\bm{\alpha}} \psi_n
&=& \alpha_{n+1} \psi_{n+1}, \;\;\; n=0,1, \ldots . \nonumber
\end{eqnarray}
\end{lemma}

This Lemma suggests to call $A_{\varphi,\psi}^{\bm{\alpha}}$ and $B_{\varphi,\psi}^{\bm{\alpha}}$ (resp. $A_{\psi,\varphi}^{\bm{\alpha}}$ and $B_{\psi,\varphi}^{\bm{\alpha}}$)  generalized lowering and raising operators for $\F_\varphi$ (resp. $\F_\psi$).

By Lemma 4.1, (1), if $D_\varphi$ is dense in $\Hil$, then the operators $A_{\varphi,\psi}^{\bm{\alpha}}$ and $B_{\varphi,\psi}^{\bm{\alpha}}$ are densely defined and $(A_{\varphi,\psi}^{\bm{\alpha}})^\ast \supseteq A_{\psi,\varphi}^{\overline{\bm{\alpha}}}$ and $(B_{\varphi,\psi}^{\bm{\alpha}})^\ast \supseteq B_{\psi,\varphi}^{\overline{\bm{\alpha}}}$. By Lemma 4.1, (2), if $D_\psi$ is dense in $\Hil$, then $A_{\psi,\varphi}^{\bm{\alpha}}$ and $B_{\psi,\varphi}^{\bm{\alpha}}$ are densely defined and $(A_{\psi,\varphi}^{\bm{\alpha}})^\ast \supseteq A_{\varphi,\psi}^{\overline{\bm{\alpha}}}$ and $(B_{\psi,\varphi}^{\bm{\alpha}})^\ast \supseteq B_{\varphi,\psi}^{\overline{\bm{\alpha}}}$. But, in case that both $D_\varphi$ and $D_\psi$ are not dense in $\Hil$, these operators are not necessarily densely defined in $\Hil$. To investigate these operators in more details, we consider the following result (\cite{hiro_taka} Lemma 2.2). \\
\par
\begin{prop} Let $(\F_{\bm{e}},T)$ be a constructing pair for a generalized Riesz system $\F_\varphi$ and $T=U|T|$ the polar decomposition of $T$. Then $ \bm{f} := \{ U e_n \}$ is an ONB in $\Hil$ and $(\F_{\bm{f}},|T|)$ is a constructing pair for $\F_\varphi$.
\end{prop}
\par
By Proposition 4.2, we can restate the notion of generalized Riesz systems in the following, more convenient, way:\\
\begin{defn} A sequence $\F_\varphi = \{ \varphi_n \}$ in $\Hil$ is said to be a generalized Riesz system if there exist an ONB $ \{ e_n \}$ in $\Hil$ and a non-singular positive self-adjoint operator $T$ in $\Hil$ such that $\{ e_n \} \subset D(T) \cap D(T^{-1})$ and $\varphi_n =Te_n$, $n=0,1, \ldots $. Then $(\F_{\bm{e}},T)$ is called a constructing pair for $\F_\varphi$ and $T$ is called a constructing operator for $\F_\varphi$.
\end{defn}
Hereafter, we assume that a constructing operator for a generalized Riesz system is a non-singular positive self-adjoint operator. Furthermore, throughout the rest of this section, let $\F_\varphi$ be a generalized Riesz system with a constructing pair $(\F_{\bm{e}},T)$ and $\psi_n =T^{-1} e_n$, $n=0,1, \ldots$. We consider when the non-self-adjoint Hamiltonians $H_{\varphi,\psi}^{\bm{\alpha}}$ and $H_{\psi,\varphi}^{\bm{\alpha}}$, the generalized lowering operators $A_{\varphi,\psi}^{\bm{\alpha}}$ and $A_{\psi,\varphi}^{\bm{\alpha}}$ and the generalized raising operators $B_{\varphi,\psi}^{\bm{\alpha}}$ and $B_{\psi,\varphi}^{\bm{\alpha}}$ are densely defined and closed operators by investigating the relations of the constructing operator $T$ and the usual self-adjoint Hamiltonian $H_{\bm{e}}^{\bm{\alpha}}$, the lowering operator $A_{\bm{e}}^{\bm{\alpha}}$ and the raising operator $B_{\bm{e}}^{\bm{\alpha}}$.\\
\par
\begin{prop} Let $X=H_{\bm{e}}^{\bm{\alpha}}$ (resp. $A_{\bm{e}}^{\bm{\alpha}}$, $B_{\bm{e}}^{\bm{\alpha}}$). The following statements hold:
\par
(1) If $D(T) \subseteq D(X)$ and $XD(T) \subseteq D(T)$, then $H_{\varphi,\psi}^{\bm{\alpha}}$ (resp. $A_{\varphi,\psi}^{\bm{\alpha}}$, $B_{\varphi,\psi}^{\bm{\alpha}}$) is densely defined, furthermore if $T^{-1}$ is bounded, then $H_{\varphi,\psi}^{\bm{\alpha}}$ (resp. $A_{\varphi,\psi}^{\bm{\alpha}}$, $B_{\varphi,\psi}^{\bm{\alpha}}$) is closed.
\par
(2) If $R(T)$ (the range of $T$)$ \subseteq D(X)$ and $XR(T) \subseteq R(T)$, then $H_{\psi,\varphi}^{\bm{\alpha}}$ (resp. $A_{\psi,\varphi}^{\bm{\alpha}}$, $B_{\psi,\varphi}^{\bm{\alpha}}$) is densely defined, furthermore if $T$ is bounded, then $H_{\psi,\varphi}^{\bm{\alpha}}$ (resp. $A_{\psi,\varphi}^{\bm{\alpha}}$, $B_{\psi,\varphi}^{\bm{\alpha}}$) is closed.\end{prop}
 \begin{proof}In case that $X=H_{\bm{e}}^{\bm{\alpha}}$, we show (1). By the assumption of (1), we have $D(H_{\varphi,\psi}^{\bm{\alpha}})=R(T)$, and since $T$ is a non-singular positive self-adjoint operator in $\Hil$, it follows that $R(T)$ is dense in $\Hil$. Hence $H_{\varphi,\psi}^{\bm{\alpha}}$ is densely defined. Suppose that $T^{-1}$ is bounded. Take an arbitrary $\{ x_n \}$ in $D(T)$ such that $\lim_{n \rightarrow \infty} T x_n =y$ and $\lim_{n \rightarrow \infty} H_{\varphi,\psi}^{\bm{\alpha}} T x_n = \lim_{n \rightarrow} TH_{\bm{e}}^{\bm{\alpha}} x_n =z$. Then, since $T^{-1}$ is bounded, we have $\lim_{n \rightarrow \infty} x_n = T^{-1} y$ and $\lim_{n \rightarrow \infty} H_{\bm{e}}^{\bm{\alpha}} x_n = T^{-1} z$, which implies that $T^{-1}y \in D(H_{\bm{e}}^{\bm{\alpha}})$ and $H_{\bm{e}}^{\bm{\alpha}}T^{-1} y= T^{-1}z$. Hence we have $y= T(T^{-1}y) \in R(T)=D(H_{\varphi,\psi}^{\bm{\alpha}})$ and $z=T(T^{-1}z)=TH_{\bm{e}}^{\bm{\alpha}}T^{-1} y= H_{\varphi,\psi}^{\bm{\alpha}} y$. Thus, $H_{\varphi,\psi}^{\bm{\alpha}}$ is a closed densely defined operator in $\Hil$. The other statements can be proved similarly.

 \end{proof}

\begin{prop} Let $X=H_{\bm{e}}^{\bm{\alpha}}$ (resp. $A_{\bm{e}}^{\bm{\alpha}}$, $B_{\bm{e}}^{\bm{\alpha}}$). Then the following statements hold:\\
\par
(1) If $D(X) \cup R(X) \subseteq D(T)$ and $TD(X)$ is dense in $\Hil$, then $D(H_{\varphi,\psi}^{\bm{\alpha}})$ (resp. $D(A_{\varphi,\psi}^{\bm{\alpha}})$, $D(B_{\varphi,\psi}^{\bm{\alpha}})$) $\supseteq TD(X)$, and so $H_{\varphi,\psi}^{\bm{\alpha}}$ (resp. $A_{\varphi,\psi}^{\bm{\alpha}}$, $B_{\varphi,\psi}^{\bm{\alpha}}$) is densely defined, furthermore if $T^{-1}$ is bounded, then $H_{\varphi,\psi}^{\bm{\alpha}}$ (resp. $A_{\varphi,\psi}^{\bm{\alpha}}$, $B_{\varphi,\psi}^{\bm{\alpha}}$) is closed.
\par
(2) If $D(X) \cup R(X) \subseteq D(T^{-1})$ and $T^{-1}D(X)$ is dense in $\Hil$, then $H_{\psi,\varphi}^{\bm{\alpha}}$ (resp. $A_{\psi,\varphi}^{\bm{\alpha}}$, $B_{\psi,\varphi}^{\bm{\alpha}}$) is densely defined, furthermore if $T$ is bounded, then $H_{\psi,\varphi}^{\bm{\alpha}}$ (resp. $A_{\psi,\varphi}^{\bm{\alpha}}$, $B_{\psi,\varphi}^{\bm{\alpha}}$) is closed.\end{prop}
The proof is analogous to that of Proposition 4.4, and will not be repeated.
\par
\begin{cor} (1) Suppose that $T$ is bounded. Then $H_{\varphi,\psi}^{\bm{\alpha}}$, $A_{\varphi,\psi}^{\bm{\alpha}}$ and $B_{\varphi,\psi}^{\bm{\alpha}}$ are densely defined.
\par
(2) Suppose that $T^{-1}$ is bounded. Then $H_{\psi,\varphi}^{\bm{\alpha}}$, $A_{\psi,\varphi}^{\bm{\alpha}}$ and $B_{\psi,\varphi}^{\bm{\alpha}}$ are densely defined.\end{cor}
 \begin{proof}Since $T$ is bounded and non-singular, it is easily shown that $TD(H_{\bm{e}}^{\bm{\alpha}})$, $TD(A_{\bm{e}}^{\bm{\alpha}})$ and $TD(B_{\bm{e}}^{\bm{\alpha}})$ are dense in $\Hil$. Hence it follows from Proposition 4.5 that $H_{\varphi,\psi}^{\bm{\alpha}}$, $A_{\varphi,\psi}^{\bm{\alpha}}$ and $B_{\varphi,\psi}^{\bm{\alpha}}$ are densely defined. Similarly we can show (2).

 \end{proof}
We devote the last part of this section to construct an algebraic structure useful in the analysis of the operators considered so far. This is on the same line as the approach discussed in \cite{bgst} for non-self-adjoint position and momentum operators. Then we introduce now the notion of unbounded operator algebras. Let $\D$ be a dense subspace in a Hilbert space $\Hil$. We denote by $\Lc(\D)$ the set of all linear operators from $\D$ to $\D$ and put
\begin{eqnarray}
\Lc^\dagger (\D)
= \{ x \in \Lc(\D) ; D(X^\ast) \supset \D \; {\rm and } \; X^\ast \D \subset \D \}. \nonumber
\end{eqnarray}
Then $\Lc (\D)$ is an algebra equipped with the usual operations: $X+Y$, $\alpha X$ and $XY$, and $\Lc^\dagger (\D)$ is a $\ast$-algebra with involution $X^\dagger := X^\ast \lceil_\D$ (the restriction of $X^\ast$ to $\D$). A $\ast$-subalgebra of $\Lc^\dagger(\D)$ is called an $O^\ast$-algebra on $\D$ \cite{schm, ait_book}.

We assume that
\begin{eqnarray}
0 \leq \alpha_0 < \alpha_n <\alpha_{n+1} \;\;\; {\rm and } \;\;\; \alpha_{n+1} \leq \alpha_n +r , \;\; n=0,1, \ldots,
\end{eqnarray}
and put
\begin{eqnarray}
\D := \cap_{n \in N} D((H_{\bm{e}}^{\bm{\alpha}})^n). \nonumber
\end{eqnarray}
Then $\D$ is a dense subspace in $\Hil$, and we have the following\\
\par
\begin{lemma} (1) $H_{\bm{e}}^{\bm{\alpha}} \lceil_\D \in \Lc^\dagger (\D)$ and $(H_{\bm{e}}^{\bm{\alpha}} \lceil_\D )^\dagger = H_{\bm{e}}^{\bm{\alpha}} \lceil_\D$.
\par
(2) $A_{\bm{e}}^{\bm{\alpha}} \lceil_\D \in \Lc^\dagger(\D)$ and $(A_{\bm{e}}^{\bm{\alpha}} \lceil_\D)^\dagger = B_{\bm{e}}^{\bm{\alpha}}\lceil_\D$.\end{lemma}
\begin{proof} Since
\begin{eqnarray}
H_{\bm{e}}^{\bm{\alpha}}x
&=& \sum_{k=0}^\infty \alpha_k <x,e_k>e_k , \;\;\; x \in D(H_{\bm{e}}^{\bm{\alpha}}) , \nonumber \\
A_{\bm{e}}^{\bm{\alpha}}x
&=& \sum_{k=0}^\infty \alpha_{k+1} <x,e_{k+1}>e_k , \;\;\; x \in D(A_{\bm{e}}^{\bm{\alpha}}) \nonumber \\
B_{\bm{e}}^{\bm{\alpha}}x
&=& \sum_{k=0}^\infty \alpha_{k+1} <x,e_k>e_{k+1}, \;\;\; x \in D(B_{\bm{e}}^{\bm{\alpha}}) , \nonumber
\end{eqnarray}
it follows from (4.5) that $D(H_{\bm{e}}^{\bm{\alpha}})=D(A_{\bm{e}}^{\bm{\alpha}})=D(B_{\bm{e}}^{\bm{\alpha}})$. The statement (1) follows from
\begin{eqnarray}
x\in \D = \cap_{n \in N}D((H_{\bm{e}}^{\bm{\alpha}})^n) \;\;\; {\rm if  \; and  \; only \; if} \; \sum_{k=0}^\infty \alpha_{2k}^{2n} |<x,e_k>|^2 < \infty , \;\; n \in N .
\end{eqnarray}
Furthermore, since
\begin{eqnarray}
(H_{\bm{e}}^{\bm{\alpha}})^n A_{\bm{e}}^{\bm{\alpha}} x
&=& \sum_{k=0}^\infty \alpha_k^n \alpha_{k+1}<x,e_{k+1}>e_k, \;\;\; x \in \D , \nonumber \\
(H_{\bm{e}}^{\bm{\alpha}})^n B_{\bm{e}}^{\bm{\alpha}} x
&=& \sum_{k=0}^\infty \alpha_{k+1}^{n+1}<x,e_k>e_{k+1}, \;\;\; x\in \D, \nonumber
\end{eqnarray}
it follows from (4.5) that
\begin{eqnarray}
\sum_{k=0}^\infty \alpha_k^{2n} \alpha_{k+1}^2 |<x,e_{k+1}>|^2
&\leq& \sum_{k=0}^\infty \alpha_{k+1}^{2(n+1)}|<x,e_{k+1}>|^2 \nonumber \\
&\leq& \sum_{k=0}^\infty \alpha_k^{2(n+1)}|<x,e_k>|^2 \nonumber
\end{eqnarray}
and
\begin{eqnarray}
\sum_{k=0}^\infty \alpha_{k+1}^{2(n+1)} |<x,e_k>|^2 \leq \sum_{k=0}^\infty (\alpha_k +r)^{2(n+1)}|<x,e_k>|^2 \nonumber
\end{eqnarray}
for all $x\in \D$ and $n \in N$, which implies by (4.6) that $A_{\bm{e}}^{\bm{\alpha}} \D \subseteq \D$ and $B_{\bm{e}}^{\bm{\alpha}} \subseteq \D$. Thus, (2) holds.

\end{proof}

By Lemma 4.7, we have the following\\
\par
\begin{prop} (1) Suppose $T\D \subseteq \D$ and $T\D$ is dense in $\Hil$. We denote by $\E$ the linear span of $T\D$. Then, $D(A_{\varphi,\psi}^{\bm{\alpha}}) \cap D(B_{\varphi,\psi}^{\bm{\alpha}}) \supseteq \E$, $A_{\varphi,\psi}^{\bm{\alpha}} \lceil_\E$, $B_{\varphi,\psi}^{\bm{\alpha}} \lceil_\E \in \Lc (\E)$ and
\begin{eqnarray}
\left( A_{\varphi,\psi}^{\bm{\alpha}} \lceil_\E \right)^m \left( B_{\varphi,\psi}^{\bm{\alpha}} \lceil_\E \right)^l
&=& T(A_{\bm{e}}^{\bm{\alpha}})^m (B_{\bm{e}}^{\bm{\alpha}})^l T^{-1} \lceil_\E , \nonumber \\
\left( B_{\varphi,\psi}^{\bm{\alpha}} \lceil_\E \right)^m \left( A_{\varphi,\psi}^{\bm{\alpha}} \lceil_\E \right)^l
&=& T(B_{\bm{e}}^{\bm{\alpha}})^m (A_{\bm{e}}^{\bm{\alpha}})^l T^{-1} \lceil_\E , \;\;\; m,l= 0,1, \ldots . \nonumber
\end{eqnarray}
(2) Suppose $T^{-1} \D \subseteq \D$ and $T^{-1} \D$ is dense in $\Hil$. We denote by $\E_-$ the linear span of $T^{-1} \D$. Then, $D(A_{\psi,\varphi}^{\bm{\alpha}}) \cap D(B_{\psi,\varphi}^{\bm{\alpha}}) \supseteq \E_-$, $A_{\psi,\varphi}^{\bm{\alpha}} \lceil_{\E_-}$, $B_{\psi,\varphi}^{\bm{\alpha}} \lceil_{\E_-} \in \Lc(\E_-)$ and
\begin{eqnarray}
\left( A_{\psi,\varphi}^{\bm{\alpha}} \lceil_{\E_-} \right)^m \left( B_{\psi,\varphi}^{\bm{\alpha}} \lceil_{\E_-} \right)^l
&=& T^{-1}(A_{\bm{e}}^{\bm{\alpha}})^m (B_{\bm{e}}^{\bm{\alpha}})^l T \lceil_{\E_-} , \nonumber \\
\left( B_{\psi,\varphi}^{\bm{\alpha}} \lceil_{\E_-} \right)^m \left( A_{\psi,\varphi}^{\bm{\alpha}} \lceil_{\E_-} \right)^l
&=& T^{-1}(B_{\bm{e}}^{\bm{\alpha}})^m (A_{\bm{e}}^{\bm{\alpha}})^l  \lceil_{\E_-} \;\;\; m,l= 0,1, \ldots . \nonumber
\end{eqnarray}
(3) Suppose that $T \D =\D$. Then the domains of the operators $A_{\varphi,\psi}^{\bm{\alpha}}$, $B_{\varphi,\psi}^{\bm{\alpha}}$, $A_{\psi,\varphi}^{\bm{\alpha}}$ and $B_{\psi,\varphi}^{\bm{\alpha}}$ contain $\D$ and the restrictions of these operators to $\D$ belong to $\Lc^\dagger (\D)$, and
\begin{eqnarray}
\left( A_{\varphi,\psi}^{\bm{\alpha}} \lceil_\D \right)^m \left( B_{\varphi,\psi}^{\bm{\alpha}} \lceil_\D \right)^l
&=& T(A_{\bm{e}}^{\bm{\alpha}})^m (B_{\bm{e}}^{\bm{\alpha}})^l T^{-1} \lceil_\D , \nonumber \\
\left( B_{\varphi,\psi}^{\bm{\alpha}} \lceil_\D \right)^m \left( A_{\varphi,\psi}^{\bm{\alpha}} \lceil_\D \right)^l
&=& T(B_{\bm{e}}^{\bm{\alpha}})^m (A_{\bm{e}}^{\bm{\alpha}})^l T^{-1} \lceil_\D , \nonumber \\
\left( A_{\psi,\varphi}^{\bm{\alpha}} \lceil_\D \right)^m \left( B_{\psi,\varphi}^{\bm{\alpha}} \lceil_\D \right)^l
&=& T^{-1}(A_{\bm{e}}^{\bm{\alpha}})^m (B_{\bm{e}}^{\bm{\alpha}})^l T \lceil_\D , \nonumber \\
\left( B_{\psi,\varphi}^{\bm{\alpha}} \lceil_\D \right)^m \left( A_{\psi,\varphi}^{\bm{\alpha}} \lceil_\D \right)^l
&=& T^{-1}(B_{\bm{e}}^{\bm{\alpha}})^m (A_{\bm{e}}^{\bm{\alpha}})^l  \lceil_\D \;\;\; m,l= 0,1, \ldots . \nonumber
\end{eqnarray}
(4) Let $\alpha_n = \sqrt{n}$, $n=0,1, \ldots $ and consider the following equations:
\begin{eqnarray}
A_{\varphi,\psi}^{\bm{\alpha}}B_{\varphi,\psi}^{\bm{\alpha}}-B_{\varphi,\psi}^{\bm{\alpha}}A_{\varphi,\psi}^{\bm{\alpha}} &=& \1, \\
A_{\psi,\varphi}^{\bm{\alpha}}B_{\psi,\varphi}^{\bm{\alpha}}-B_{\psi,\varphi}^{\bm{\alpha}}A_{\psi,\varphi}^{\bm{\alpha}} &=& \1.
\end{eqnarray}
Then, (4.7) (resp. (4.8)) holds on $\E$ (resp. $\E_-$) under the assumption in (1) (resp. (2)), and both (4.7) and (4.8) hold on $\D$ under the assumption in (3).\end{prop}
 \begin{proof}The statement (1) and (2) follow immediately from Lemma 4.7. The statement (3) follows from (1) and (2), and (4) follows from $A_{\bm{e}}^{\bm{\alpha}}B_{\bm{e}}^{\bm{\alpha}}-B_{\bm{e}}^{\bm{\alpha}}A_{\bm{e}}^{\bm{\alpha}} = \1$ on $\D$.

 \end{proof}

{\bf Remark.} Algebraic operators of the generalized lowering and raising operators $A_{\varphi,\psi}^{\bm{\alpha}}$ and $B_{\varphi,\psi}^{\bm{\alpha}}$ (resp. $A_{\psi,\varphi}^{\bm{\alpha}}$ and $B_{\psi,\varphi}^{\bm{\alpha}}$) for $\F_\varphi$ (resp. $\F_\psi$) are determined by those of the usual lowering and raising operators $A_{\bm{e}}^{\bm{\alpha}}$ and $B_{\bm{e}}^{\bm{\alpha}}= (A_{\bm{e}}^{\bm{\alpha}})^\ast$. But, with respect to the algebraic operations of the operators for $\F_\varphi$ and the operator for $\F_\psi$ this is not true. For example, under the assumption in (3), for the multiplication of $A_{\psi,\varphi}^{\bm{\alpha}}$ and $B_{\varphi,\psi}^{\bm{\alpha}}$ we have,
\begin{eqnarray}
A_{\psi,\varphi}^{\bm{\alpha}}B_{\varphi,\psi}^{\bm{\alpha}}= T^{-1} A_{\bm{e}}^{\bm{\alpha}} T^2 B_{\bm{e}}^{\bm{\alpha}} T^{-1} \nonumber
\end{eqnarray}
on $\D$.

\section{Conclusions}\label{sectconl}

{As we have seen, generalized Riesz systems, discussed in this paper share with {\em true} Riesz bases a series of interesting properties, whose nature is independent of the fact that Riesz bases are actually a frame; this property is indeed missing in our framework, because of the unboundedness of the operators $T$, $T^{-1}$ that link a generalized Riesz system with an orthonormal basis of $\Hil$. The crucial assumption we have made in this paper is that a generalized Riesz system $\F_\varphi$ and its biorthogonal dual system $\F_\psi$ constitute a $\D$-quasi basis; that is the equality \eqref{eqn_quasi_bases} holds on $\D$; this puts on the stage sesquilinear forms of the type
$$
 \Omega_{\varphi,\psi}(x,y)=\sum_{n=0}^\infty\ip{x}{\varphi_n}\ip{\psi_n}{y}, \quad x,y \in \D.
$$
For $\D$-quasi bases they may exhibit a {\em singular} behavior. In this case, in fact, the equality $\Omega_{\varphi,\psi}(x,y)=\ip{x}{y}$ for all $x,y\in \D$ shows that  $\Omega_{\varphi,\psi}$ extends everywhere in $\Hil \times \Hil$ to the inner product of $\Hil$, but the convergence of the series defining it is not guaranteed in $\Hil \times \Hil$. This is just one of the themes that should be investigated in this respect, together with concrete physical systems where this singular aspect, and operators as those considered in Section 4, really matter and have a precise meaning. We hope to consider these questions in future papers.

}

\section*{Acknowledgements}
This work was partially supported by the University of Palermo, by the Gruppo Nazionale per la Fisica Matematica (GNFM) and by the Gruppo Nazionale per l'Analisi Matematica, la
Probabilit\`{a} e le loro Applicazioni (GNAMPA) of the Istituto
Nazionale di Alta Matematica (INdAM). The authors thank Prof. A. Inoue for his valuable comments on the paper.

\end{document}